\documentclass[aps,pra,notitlepage,twocolumn,superscriptaddress,nofootinbib]{revtex4-2}
\usepackage{amsmath,amssymb,amstext}
\usepackage{amsthm}
\usepackage{bm}
\usepackage{bbm}
\usepackage{blochsphere}
\usepackage{braket, physics}
\usepackage{comment}
\usepackage{enumerate}
\usepackage{framed}
\usepackage{graphicx}
\usepackage[colorlinks,linkcolor=blue,anchorcolor=blue,citecolor=blue,urlcolor=blue]{hyperref}
\usepackage{ifthen}
\usepackage[utf8]{inputenc}
\usepackage{listliketab}
\usepackage{lmodern}
\usepackage{mathrsfs}
\usepackage{mathtools}
\usepackage{natbib}
\usepackage{setspace}
\usepackage{standalone}
\usepackage[caption=false,labelformat=simple]{subfig}
\usepackage{svg}
\usepackage{tikz}
\usepackage{tikz-3dplot}
\usepackage{xcolor}

\everymath{\displaystyle}
\newtheorem{theorem}{Theorem}
\newtheorem{lemma}{Lemma}
\newtheorem{proposition}{Proposition}
\newcommand*{\blank}{\makebox[\widthof{$\;\coloneqq\;$}]{$\cdot$}}
\newcommand{\thm}[1]{\hyperref[thm:#1]{Theorem~\ref*{thm:#1}}}
\newcommand{\lem}[1]{\hyperref[lem:#1]{Lemma~\ref*{lem:#1}}}
\newcommand{\prop}[1]{\hyperref[prop:#1]{Proposition~\ref*{prop:#1}}}
\newcommand{\fig}[1]{\hyperref[fig:#1]{FIG.~\ref*{fig:#1}}}
\newcommand{\figg}[2]{\hyperref[fig:#1]{FIG.~\ref*{fig:#1}(#2)}}
\renewcommand{\sec}[1]{\hyperref[sec:#1]{Sec.~\ref*{sec:#1}}}

\DeclareMathOperator{\poly}{poly}
\DeclareMathOperator{\diag}{diag}
\usetikzlibrary{quantikz}
\tikzset{operator/.append style={rounded corners}}

\begin{document}

\preprint{APS/123-QED}

\title{Digital quantum simulator for the time-dependent Dirac equation using \texorpdfstring{\\}{} discrete-time quantum walks}

\author{Shigetora Miyashita}
\affiliation{Faculty of Environment and Information Studies, Keio University, 5322 Endo, Fujisawa, Kanagawa 252-0882, Japan}
\author{Takahiko Satoh}
\affiliation{Graduate School of Science and Technology, Keio University,
3-14-1 Hiyoshi, Kohoku-ku, Yokohama 223-8522, Japan}
\affiliation{Quantum Computing Center, Keio University, 3-14-1 Hiyoshi, Kohoku-ku, Yokohama 223-8522, Japan}
\author{Michihiko Sugawara}
\affiliation{Graduate School of Science and Technology, Keio University,
3-14-1 Hiyoshi, Kohoku-ku, Yokohama 223-8522, Japan}
\affiliation{Quantum Computing Center, Keio University, 3-14-1 Hiyoshi, Kohoku-ku, Yokohama 223-8522, Japan}
\author{Naphan Benchasattabuse}
\affiliation{Graduate School of Media and Governance, Keio University, 5322 Endo, Fujisawa, Kanagawa 252-0882, Japan}
\author{Ken M. Nakanishi}
\affiliation{Institute for Physics of Intelligence, The University of Tokyo, 7-3-1 Hongo, Bunkyo-ku, Tokyo 113-0033, Japan}
\author{Michal Hajdu\v{s}ek}
\affiliation{Graduate School of Media and Governance, Keio University, 5322 Endo, Fujisawa, Kanagawa 252-0882, Japan}
\author{Hyensoo Choi}
\affiliation{Faculty of Environment and Information Studies, Keio University, 5322 Endo, Fujisawa, Kanagawa 252-0882, Japan}
\author{Rodney Van Meter}
\affiliation{Faculty of Environment and Information Studies, Keio University, 5322 Endo, Fujisawa, Kanagawa 252-0882, Japan}

\date{\today}

\begin{abstract}
    We introduce a quantum algorithm for simulating the time-dependent Dirac equation in 3+1 dimensions using discrete-time quantum walks. Thus far, promising quantum algorithms have been proposed to simulate quantum dynamics in non-relativistic regimes efficiently. However, only some studies have attempted to simulate relativistic dynamics due to its theoretical and computational difficulty. By leveraging the convergence of discrete-time quantum walks to the Dirac equation, we develop a quantum spectral method that approximates smooth solutions with exponential convergence. This mitigates errors in implementing potential functions and reduces the overall gate complexity that depends on errors. We demonstrate that our approach does not require additional operations compared to the asymptotic gate complexity of non-relativistic real-space algorithms. Our findings indicate that simulating relativistic dynamics is achievable with quantum computers and can provide insights into relativistic quantum physics and chemistry.
\end{abstract}

\maketitle

\section{Introduction}
\label{sec:dtqw:introduction}

The Dirac equation~\cite{thaller2013dirac} successfully combines quantum mechanics with special relativity to accurately describe the behavior of electrons and positrons at relativistic speeds. This equation predicts unique properties of Dirac fermions that arise from the natural incorporation of relativistic effects. These effects include Zitterbewegung, which causes trembling motion, and the Klein paradox, which can be interpreted using particle-antiparticle production~\cite{zbMATH02565649,PhysRevD.23.2454,Hestenes1990-ld,doi:10.1119/1.1933296,Sidharth_2008,PhysRevLett.100.153002,Katsnelson2006-dw,PhysRevA.82.020101}. The Schrödinger equation cannot account for these effects. The incorporation of such corrections benefits various areas of physics. In particle physics, theoretical studies have shown that Dirac particles act as standard model fermions~\cite{cottingham_greenwood_2007} and may be related to neutrinos~\cite{Bilenky:2020wjn,BALANTEKIN2019488}. Additionally, the Dirac operator is critical in lattice gauge theory~\cite{PhysRevD.10.2445,PhysRevD.11.395,Martinez_2016,Bañuls2020,Dalmonte_2016,PRXQuantum.4.027001}. It has also been found that fermions appear as Dirac semimetals in condensed matter physics~\cite{PhysRevB.85.195320,PhysRevB.88.125427,PhysRevLett.108.140405}. Finally, the Dirac equation precisely predicts the electron orbits of heavy elements in chemistry by decoupling positron states~\cite{https://doi.org/10.1002/bbpc.19971010102,doi:10.1146/annurev.pc.36.100185.002203}.

Previous research has shown that discrete-time quantum walks~\cite{PhysRevA.58.915,WATROUS2001376,PhysRevA.70.042312} are essential models for understanding the time-dependent Dirac equation~\cite{thaller2013dirac,Gerritsma_2010,arrighi2014dirac,PhysRevLett.98.253005,Fleury_2023}. For instance, Strauch focused on the relationship between the one-dimensional Dirac equation and quantum walks~\cite{PhysRevA.73.054302}. He observed that discrete-time quantum walkers exhibit relativistic quantum dynamics while continuous-time quantum walkers exhibit non-relativistic quantum dynamics. Moreover, Chandrashekar \emph{et al.} recast the quantum walk operators (coin and shift) as the time evolution operators in relativistic quantum mechanics (mass and kinetic term) by taking the continuum limits~\cite{PhysRevA.81.062340}. This study proposed extending quantum walks in curved spacetime~\cite{PhysRevA.88.042301,Arrighi2016,mallick2019simulating}.

There has been significant progress in numerical simulations of the time-dependent Dirac equation~\cite{PhysRevA.59.604,FILLIONGOURDEAU20121403,Fillion_Gourdeau_2014,PhysRevLett.54.669,PhysRevA.53.1605,LORIN2022108474}, but simulation protocols still need improvement. Unless formulated as two-dimensional graphene-like matter, Dirac fermions evolve in 3+1 dimensions, as represented by heavy-ion scattering~\cite{PhysRevLett.55.2786,PhysRevA.24.103}. However, simulating the 3+1 dimensional Dirac equation generally requires a supercomputer, even with advanced algorithms, as the memory required for representing quantum systems scale exponentially with the system size~\cite{ANTOINE2017150}. To address this issue, digital quantum simulators based on quantum walks for the Dirac equation have gained attention~\cite{Fillion_Gourdeau_2017,Huerta_Alderete2020-de}. These simulators reduce the memory size using qubits, as well as the time complexity. However, implementing these approaches is limited by the number of quantum gates and auxiliary qubits. This is due to the frequent use of expensive multi-controlled CNOT gates~\cite{PhysRevA.106.042602}.

We address this problem by introducing a novel algorithm for simulating the 3+1 dimensional time-dependent Dirac equation using quantum walks.
In order to avoid the use of multi-controlled CNOT gates, we employ quantum spectral methods~\cite{Childs_2020,Childs2021highprecision,Childs2022quantumsimulationof} that diagonalize the kinetic term in Fourier space.
These methods allow us to approximate the solution as a sum of few basis functions while simultaneously maintaining high accuracy.
In addition, we describe the implementation of potential terms using position-dependent coin and phase operators~\cite{Ahmad_2020,nzongani2022quantum}, which allow for the introduction of arbitrarily shaped scalar and vector potentials of U(1) gauge theory. We compare the resources required to implement the quantum algorithms with non-relativistic cases. A detailed analysis shows that the asymptotic gate complexity of the proposed algorithms for the Dirac equation scales identically to the Schr\"{o}dinger equation for the same potentials. Finally, we apply the algorithms to simulate Zitterbewegung in both 1+1 and 3+1 dimensions and the Klein paradox in 1+1 dimensions. Due to the nature of spectral methods, we demonstrate that spectral convergence in our scheme is guaranteed.

The following is an outline of our paper. In \sec{dtqw:the Dirac equation}, we define the time-dependent Dirac equation in 3+1 dimensions. To simulate the dynamics on a quantum computer, we first perform Trotterization (time discretization) in \sec{dtqw:Time discretization} to decompose the Dirac Hamiltonian. We then discretize space in \sec{dtqw:Space discretization} to introduce computational grids. \sec{dtqw:discrete-time quantum walk} shows our implementation of the time evolution operator based on discrete-time quantum walks. We summarize the procedure by term by explaining the relationship between the time evolution operator and quantum walk operators. In \sec{dtqw:simulation results}, we present classical and quantum simulations of the Dirac equation. The simulations are tested in 1+1 and 3+1 dimensions, and we conduct a numerical analysis of computational costs and errors. Finally, we discuss the simulation results and their potential applications in \sec{dtqw:conclusions}.

\section{Time-Dependent Dirac equation}
\label{sec:dtqw:the Dirac equation}

This section covers the preliminaries required to simulate the time-dependent Dirac equation on a digital computer.
In this paper, we consider the Dirac equation in 3+1 dimensions in natural units
\begin{align}
    \label{eq:Dirac Eq}
    \bigl[\gamma^\mu(i\partial_\mu-eA_\mu)-m\bigr]\psi=0,
\end{align}
where $m$ is the mass of a Dirac particle, $\gamma^\mu$ is a set of $4\times4$ gamma matrices, and $A_\mu$ is the minimally coupled gauge field with $\mu=0,1,2,3$. We sum over repeated indices according to the Einstein notation. It defines the four-dimensional spacetime with the anticommutator,
\begin{align}
    \label{eq:Dirac Mat}
    \gamma^\mu\gamma^\nu+\gamma^\nu\gamma^\mu=2\eta^{\mu\nu},
\end{align}
where $\eta^{\mu\nu}=\diag(+1,-1,-1,-1)$ is the Minkowski metric that characterizes the Minkowski spacetime with a non-degenerate scalar product upon a metric space $\mathbb{R}^{1,3}$. The solution to~(\ref{eq:Dirac Eq}) is given by the four-component Dirac spinor $\psi$,
\begin{align}
    \psi=
    \left(
    \begin{array}{c}
         \psi_1 \\
         \psi_2 \\
         \psi_3 \\
         \psi_4
    \end{array}
    \right).
\end{align}
The large components $(\psi_1,\psi_2)^T$ and small components $(\psi_3,\psi_4)^T$ depend on each other while the set of two components has two degrees of freedom corresponding to the spin up and down. Thus the Dirac equation is a set of coupled partial differential equations for the spinor components.

To simulate the time-dependent Dirac equation, we must discretize the Minkowski spacetime defined by $\eta^{\mu\nu}.$  In the following, we introduce the Hamiltonian form of the Dirac equation to derive the time evolution operator in a closed form. We then consider the discretization of position and momentum.

\subsection{Time discretization}
\label{sec:dtqw:Time discretization}
We now specify how time needs to be discretized in order to implement the time evolution on a discrete quantum simulator.
In the rest of this manuscript, we work in the Schr\"{o}dinger picture, where the time evolution of a quantum state $\psi$ is generated by the Hamiltonian $H$.
The time-dependent Dirac equation in the Schr\"{o}dinger picture reads~\cite{thaller2013dirac}
\begin{align}
    \label{eq:Schroedinger}
    i\frac{\partial}{\partial t}\psi=H\psi,
\end{align}
where the Dirac Hamiltonian is given by
\begin{align}
    \label{eq:Dirac_Hamiltonian}
    H=\bm{\alpha}\cdot(-i\nabla+e\bm{A})+\beta m-e\phi.
\end{align}
The Dirac matrices 
\begin{align}
    \begin{split}
        \alpha^i&= \sigma^1\otimes\sigma^i,\\
        \beta&=\sigma^3\otimes\mathbb{I}_2,
    \end{split}
\end{align}
arise from the factorization of the relativistic dispersion relation $E=\sqrt{\bm{p}^2+m^2}$ and replace $\gamma^\mu$. Note that we consider the standard Dirac representation. For the superscript $i=1,2,3$, the set of $2\times2$ Pauli spin matrices is given by
\begin{alignat}{3}
    \sigma^1=&  &{}\ket{0}\bra{1}{}&{}+{}& &\ket{1}\bra{0},\nonumber\\
    \sigma^2=&-i&{}\ket{0}\bra{1}{}&{}+{}&i&\ket{1}\bra{0},\\
    \sigma^3=&  &{}\ket{0}\bra{0}{}&{}-{}& &\ket{1}\bra{1}.\nonumber
\end{alignat}
In addition, $\bm{A}$ and $\phi$ act for potentials. While we can add the external potential $\bm{A}$ by minimal coupling, internal potential can be defined as either a zeroth component of $A_\mu$, scalar $\sim \phi$, or pseudoscalar $\sim\gamma^0\gamma^iV$ term~\cite{PhysRevA.59.1056}. We consider it as a scalar term $\phi$ hereafter.

By solving~(\ref{eq:Schroedinger}) for a propagation time $t$, one obtains the formal solution,
\begin{align}
    \label{eq:psi}
    \psi(t) = \mathcal{T}\left[e^{-i\int^t_0 H(s) d s}\right] \psi(0),
\end{align}
where $\mathcal{T}$ is the time-ordering operator.
Since~(\ref{eq:Dirac_Hamiltonian}) contains non-commuting terms, we approximate the evolution generated by the Hamiltonian by the $2k$th-order product formula~\cite{Childs2022quantumsimulationof},
with
\begin{equation}
    \begin{aligned}
        U_2(t) &\coloneqq e^{-i\frac{t}{2} H_1 } \cdot e^{-i\frac{t}{2} H_2 } \cdot e^{-i\frac{t}{2} H_3 } \cdot e^{-i\frac{t}{2} H_4 } \\ 
        & \blank e^{-i\frac{t}{2} H_4 } \cdot e^{-i\frac{t}{2} H_3 } \cdot e^{-i\frac{t}{2} H_2 } \cdot e^{-i\frac{t}{2} H_1 }, \\
        U_{2k}(t) &\coloneqq U_{2k-2}(u_kt)^2U_{2k-2}((1-4u_k)t)U_{2k-2}(u_kt)^2,
    \end{aligned}
\end{equation}
where
\begin{align}\label{eq:dtqw:diracham}
    \begin{split}
        H_1&=-i\bm{\alpha}\cdot\nabla, \\
        H_2&=\beta m, \\
        H_3&=-e\phi, \\
        H_4&=e\bm{\alpha}\cdot \bm{A},
    \end{split}
\end{align}
with $u_k\coloneqq1/(1-4^{1/(2k-1)})$. The non-commuting error decreases as the order increases. To approximate the time evolution using the $2k$th-order product formula with an error no greater than $\varepsilon/2$, the number of exponentials required can be determined using the following expression, 
\begin{align}
    \widetilde{\mathcal{O}}\Bigl( 5^{2k} (\norm{H}T)^{1+1/2k}/\varepsilon^{1/2k} \Bigr),
\end{align}
where $\norm{\cdot}$ denotes the spectral norm, and $T$ is the total evolution time. Note that the notation $\widetilde{\mathcal{O}}$ omits the poly-logarithmic complexity i.e. $\widetilde{\mathcal{O}}(g)=\mathcal{O}(g\poly(\log g))$. The detailed analysis can be found in Ref.~\cite{PhysRevX.11.011020}.
We finally obtain the discrete-time evolution by breaking the total evolution time $T$ into $r$ steps, each of size $t$,
\begin{align}
    U(T) \approx \left[ U_{2k}(t) \right]^r, \quad \mathrm{with\ }\quad r = T/t.
\end{align}

\subsection{Space discretization}
\label{sec:dtqw:Space discretization}
Having described how time can be discretized in the Schr\"{o}dinger picture, we now turn our attention to the discretization of space.
Furthermore, the Dirac spinor $\psi$ must be discretized to store quantum states on a digital computer.
Unlike continuous variables, discrete variables will have a one-to-one correspondence with the integers $\mathbb{Z}$. Here, we define $n$ as the number of grid points corresponding to the resolution of space in one dimension. Please note that $n$ is not the number of qubits. Since we consider the Dirac spinor $\psi$ in 3+1 dimensions as the subject of discretization, it has $4n^3$ grid points. Given $P_0$ type elements, piecewise constant discontinuous finite elements~\cite{LORIN2011190}, the finite volumes become equally sized cubes, and we obtain the discrete quantum states as tensor products of one-dimensional states expressed as
\begin{align}
    \psi=\sum^4_{a=1}\sum^n_{j=1}\sum^n_{k=1}\sum^n_{\ell=1}\alpha_{a,j,k,\ell}\ket{a}\ket{j}\ket{k}\ket{\ell},
\end{align}
with the normalization condition
\begin{align}
    \int d^3\bm{x}\psi^\dagger\psi=1.
\end{align}
Accordingly, we denote the position and the momentum grids with the computational cell volume $\Omega$ as
\begin{equation}
    r_p = \frac{p\,\Omega}{n}, \quad\mathrm{and\ }\quad k_p =\frac{2\pi p}{\Omega},
\end{equation}
where $p \in [-n/2, n/2) \subset\mathbb{Z}$.

Once the position and momentum grids are set, Courant--Friedrichs--Lewy (CFL) condition~\cite{10.5555/2430727} determines the time step size $t$,
\begin{align}
    t=\frac{n^\ast\Omega}{n}, \quad\mathrm{with\ }\quad n^\ast\in\frac{1}{2}\mathbb{N},
\end{align}
where $n^\ast\in\{1/2,1,3/2,2\ldots\}$ can be any (half-)integer. From a practical perspective, one would like to choose the smallest (half-)integer, known as the staggered mesh, so that the number of recursions becomes minimum.

\section{Discrete-time quantum walk}
\label{sec:dtqw:discrete-time quantum walk}

\begin{figure*}[t]
    \centering
    \includegraphics[width=\textwidth]{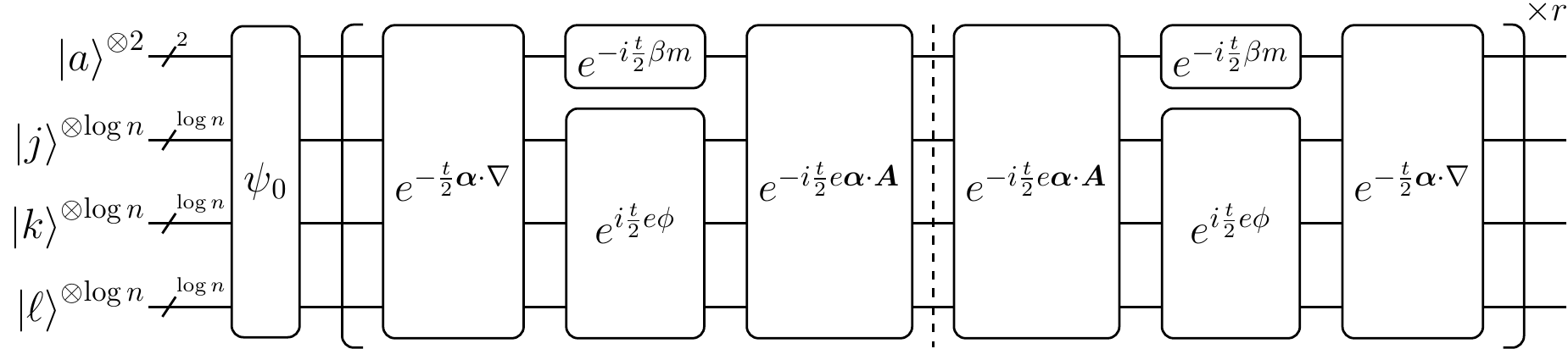}
    \caption{Quantum circuit for simulating the time-dependent Dirac equation in 3+1 dimensions using discrete-time quantum walks. The circuit consists of an initial state preparation block and $r$ repeating blocks that evaluate the time-splitting operators in the 2nd-order.}
    \label{fig:dtqw:overview}
\end{figure*}

This section shows the quantum algorithms for simulating the 3+1 dimensional time-dependent Dirac equation defined in \sec{dtqw:the Dirac equation}. As discussed in earlier studies, including~\cite{shikano2013discrete}, there are two distinct variations of quantum walk models: continuous-time and discrete-time quantum walks. While the former model quantizes the continuous-time Markov chain, the latter is based on the discrete-time Markov chain~\cite{Childs2010-dv}. Depending on the continuity, the models converge to the Schr\"{o}dinger or the Dirac equation~\cite{PhysRevA.73.054302}. We deal with the discrete-time quantum walk to our target of the time-dependent Dirac equation in 3+1 dimensions, giving the following result.

\begin{theorem}[Informal version of \thm{Theorem}]
    Consider the discrete-time quantum walk on a cycle with the relativistic Hamiltonian defined in (\ref{eq:Dirac_Hamiltonian}). Then for any $t\in\mathbb{R}$ (propagation time) and $\varepsilon$ (error tolerance) there exists a unitary operation $e^{-itH}$, which can be implemented on quantum computers with gate complexity
    \begin{align}
        \widetilde{\mathcal{O}}\Bigl( 5^{2k} (\norm{\phi}_{\max}+\norm{\bm{A}}_{\max})^{1+1/2k}T^{1+1/2k}/\varepsilon^{1/2k} \Bigr).
    \end{align}
\end{theorem}

The starting point of the implementation is constructing single-step time evolution operators. Hereafter, we employ the 1st-order product formula to implement the kinetic term and the vector potential term, which contain non-commuting terms. On the other hand, we consider the 2nd-order scheme for the total evolution because the 3+1 dimensional Dirac equation consists of 7 non-commuting terms. \fig{dtqw:overview} depicts the quantum circuit for simulating the time-dependent Dirac equation with the 2nd-order product formula. In each step, the initial quantum state $\psi_0$ is evolved by conditional shift (kinetic term $e^{-\frac{t}{2}\bm{\alpha}\cdot\nabla}$), coins (mass term $e^{-i\frac{t}{2}\beta m}$ and vector term $e^{-i\frac{t}{2}e\bm{\alpha}\cdot \bm{A}}$), and phase (scalar term $e^{i\frac{t}{2}e\phi}$). We apply the same operations but in reversed order and obtain the desired state in the 2nd order after the $r$ time steps. The following subsections explain the quantum circuit implementation of each time evolution operator.

\subsection{Shift operator as the kinetic term} \label{sec:kinetic_term}
\begin{figure*}[t]
    \centering
    \subfloat[]{\includegraphics[width=\textwidth]{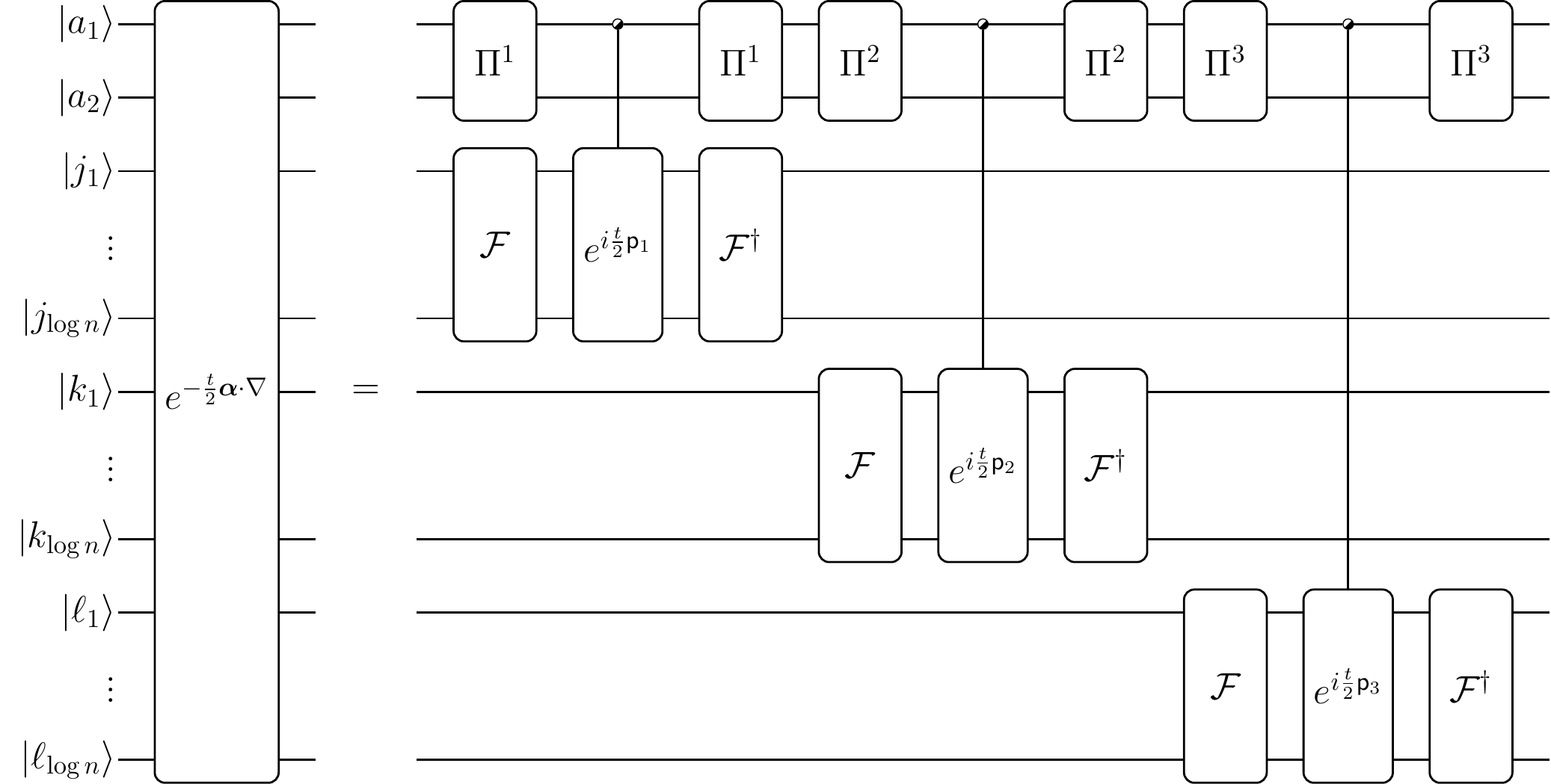}} \\
    \subfloat[]{\includegraphics[width=\textwidth]{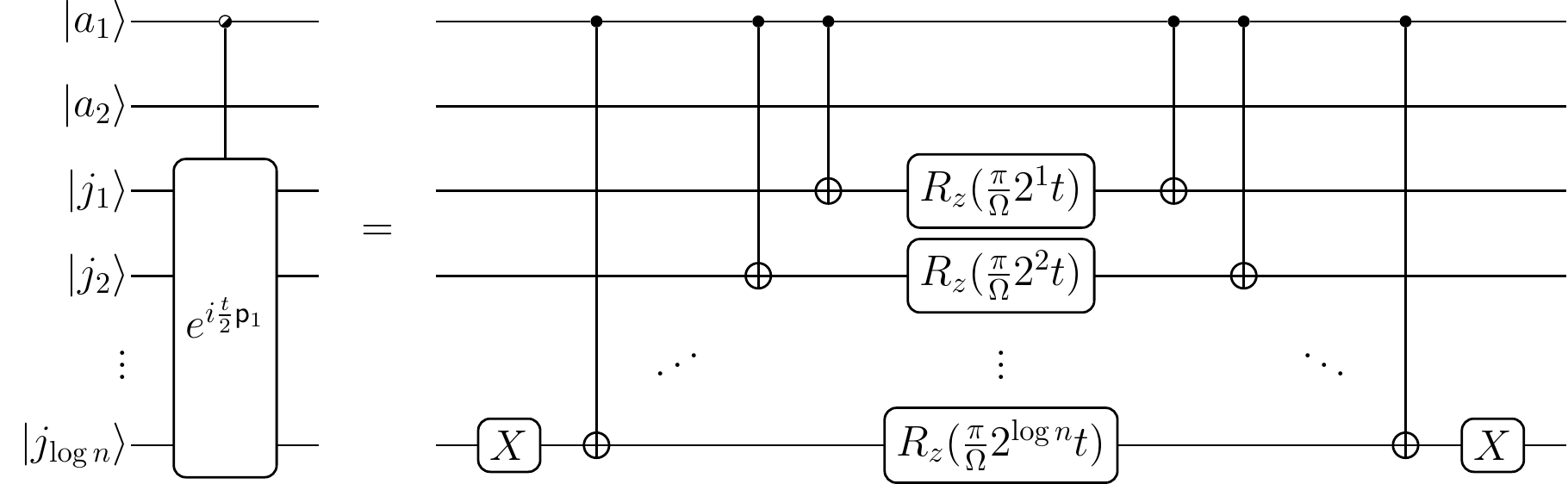}}
    \caption{(a) Quantum circuit for implementing the kinetic term of the time-dependent Dirac equation. We apply the set of $\Pi^i$ gates and quantum Fourier transforms $\mathcal{F}$ pair to efficiently implement remaining diagonal terms in the $\beta$ basis. (b) Schematics for implementing the diagonal terms. The sequence of $R_z$ gates serves as the momentum operator, and the CNOT pairs project it to the spinor space. We place single-qubit $X$ gates at the beginning and end of the sequence so that it shifts the momentum spectra.}
    \label{fig:dtqw:linear}
\end{figure*}
This subsection explains the quantum circuit implementation for the shift operator of the discrete-time quantum walk. Since the shift operation is the classical counterpart of the kinetic term~\cite{PhysRevA.81.062340}, it conditionally moves a Dirac particle to the left and right in 1+1 dimensions. The shift operation has been implemented using multi-controlled CNOT gates with stair-like configuration in the position space~\cite{Fillion_Gourdeau_2017,Acasiete_2020,Qiang_2016}. The scheme takes advantage of the kinetic term that becomes the conditional shift matrices when a time step does not exist. However, to implement the kinetic term with arbitrary time steps, we propose an implementation based on quantum spectral methods, in which we evaluate the term in the Fourier space.
\begin{lemma}\label{lem:lemma1}
    For any $t\in\mathbb{R}$, the quantum operation $e^{-\frac{t}{2}\bm{\alpha}\cdot\nabla}$ can be approximated within error tolerance $\varepsilon$ and with gate complexity $\mathcal{O}(\log{n}\log\log{n})$.
\end{lemma}
\begin{proof}
    We decompose the three-dimensional kinetic term via the product formula, as each direction does not commute. Decomposed kinetic terms introduce non-commuting error up to $\varepsilon$ in the 1st-order,
    \begin{align}
        \label{eq:Trotterized_Kinetic}
        e^{-\frac{t}{2}\bm{\alpha}\cdot\nabla} = e^{-\frac{t}{2}\alpha^3\cdot\partial_3}\cdot e^{-\frac{t}{2}\alpha^2\cdot\partial_2}\cdot e^{-\frac{t}{2}\alpha^1\cdot\partial_1} + \mathcal{O}(t^2/r^2).
    \end{align}
    Note that we can take $r=\mathcal{O}((\norm{H}t)^2/\varepsilon)$ to ensure the error at most $\varepsilon$. The right-hand side of~(\ref{eq:Trotterized_Kinetic}) contains three one-dimensional kinetic terms. From now on, we take the $x$ direction $e^{-\frac{t}{2}\alpha^1\cdot\partial_1}$ for instance, but it is defined in the same way for $i=2,3$. We then perform the eigenvalue decomposition of the Dirac matrices $\alpha^1$ and circulant matrices $e^{-t\partial_1}$ to implement them in the computational basis. We define the following unitary operators for the decomposition of Dirac matrices~\cite{Fillion_Gourdeau_2017},
    \begin{align}
        \label{eq:Pi_Matrix}
        \Pi^1=\frac{1}{\sqrt{2}}(\beta+\alpha^1).
    \end{align}
    Moreover, the circulant matrices can be diagonalized by a Fourier transform pair~\cite{davis1979circulant}. Thus we decompose the overall three-dimensional kinetic terms $e^{-\frac{t}{2}\alpha^1\cdot\partial_1}$ using the identities,
    \begin{align}
        e^{-\frac{t}{2}\alpha^1\cdot\partial_1}=(\Pi^1\otimes\mathcal{F}^\dagger)\cdot e^{i\frac{t}{2}\beta\cdot\mathsf{p}_1}\cdot(\Pi^1\otimes\mathcal{F}),
    \end{align}
    where $\mathsf{p}_1\coloneqq -i\partial/\partial x$ is the momentum operator, and $\mathcal{F}$ is defined as the quantum Fourier transform, which has gate complexity $\mathcal{O}(\log{n}\log\log{n})$~\cite{Childs2021highprecision}.
    
    We now represent the momentum operator $\mathsf{p}_1$ with a set of quantum gates. Recall that the momentum operator acts linearly. We can rewrite the operator with $\mathcal{O}(\log{n})$ sequences of $Z$ gates as follows~\cite{nielsen_chuang_2010},
    \begin{align}
        e^{i\frac{t}{2}\mathsf{p}_1}=\prod_{j=1}^{\log{n}} e^{-i2^j\pi tZ_j/\Omega}.
    \end{align}
    
    The final step is to realize the conditional shift. Since the decomposed kinetic term $e^{i\frac{t}{2}\beta\cdot\mathsf{p}_1}$ contains $\beta$ matrices concerning the large and small components of a Dirac spinor. An efficient way to impose the conditional operation is to deploy the CNOT gates before and after the momentum operations.
    
    The explicit quantum circuit implementation for the decomposed kinetic term can be seen in \fig{dtqw:linear}. To derive the time-evolution operator, we exponentiated the momentum operator, which consists of $R_z$ rotations and CNOT gates. The $\mathcal{O}(\log{n})$ CNOT gates were used to interpose single qubit rotations and add degrees of freedom~\cite{Zhou2011}. It should be noted that we shifted the spectrum of the momentum operator to the center with two single $X$ gates at the beginning and end of the shift operator~\cite{PhysRevLett.125.260511}. This leads to the gate complexity $\mathcal{O}(\log{n}\log\log{n})$.
    \end{proof}

\subsection{Coin operator as the mass term}
This subsection explains the implementation of the mass term.
\begin{lemma}\label{lem:lemma2}
    For any $t\in\mathbb{R}$, the quantum operation $e^{-i\frac{t}{2}\beta m}$ can be performed with gate complexity $\mathcal{O}(1)$.
\end{lemma}
\begin{proof}
    The mass term can be seen as a coin operator that only acts on an auxiliary qubit,
    \begin{align}
        e^{-i\frac{t}{2}\beta m}=e^{-i\frac{t}{2}mZ}.
    \end{align}
    This operation rotates the quantum states around the $Z$ axis. The term only contains the mass, a constant, with a Dirac matrix $\beta$. Thus, the mass term only contributes to the most significant qubit using $R_z$ gate~\cite{Fillion_Gourdeau_2017},
    \begin{equation}
        \begin{minipage}[c]{.4\textwidth}
            \centering
            \begin{quantikz}[draw]
                \lstick{$e^{-i\frac{t}{2}mZ}\equiv\ket{a_1}$} & \gate{R_z(tm)} & \qw
            \end{quantikz}.
        \end{minipage}
    \end{equation}
    While it is a single-qubit rotation, it runs on all the qubits simultaneously. This is because when the rotation is applied to the auxiliary register $\ket{a}$, the rotation is applied to all the quantum registers.
\end{proof}

\subsection{Position-dependent phase operator as the scalar potential term}
This subsection explains the implementation of the space-dependent scalar potential term. Before constructing specific quantum circuits corresponding to the scalar (or vector) potential terms, we briefly review the implementation of a diagonal unitary operator $e^{if}$ with a space-dependent function $f(x)$.
\begin{proposition}[{\cite{welch_2014}}]
    For a $n\times n$ diagonal unitary operator $U=e^{if}$ with a space-dependent infinitely differentiable function $f$, there exists a quantum operation that approximates $U$ with gate complexity $\mathcal{O}(\poly(\log(1/\varepsilon)))$.
\end{proposition}
\begin{proof}
    To approximate the diagonal unitary operator up to $\varepsilon$, only $n=\poly(1/\varepsilon)$ Walsh operators need to be implemented~\cite{welch_2014}, resulting in a gate complexity of one-qubit and two-qubit quantum gates of $\mathcal{O}(\poly(1/\varepsilon))$. Moreover, suppose the solution belongs to the class of infinitely differentiable functions. In that case, the spectral method can approximate the solution to an error of $\varepsilon$ using only a polynomial number of grid points, specifically $n=\poly(\log(1/\varepsilon))$~\cite{10.5555/2073480}. This leads to the gate complexity $\mathcal{O}(\poly(\log(1/\varepsilon)))$.
\end{proof}
We now provide the quantum circuit for the space-dependent scalar potential term.
\begin{lemma}\label{lem:lemma3}
    For any $t\in\mathbb{R}$, the quantum operation $e^{i\frac{t}{2}e\phi}$ can be approximated within error tolerance $\varepsilon$ and with gate complexity $\mathcal{O}(\poly(\log(1/\varepsilon)))$.
\end{lemma}
\begin{proof}
    Since the scalar potential term does not contain Dirac matrices, we can decompose it as the tensor products of position-dependent phase operators,
    \begin{align}
        e^{i\frac{t}{2}e\phi}=e^{i\frac{t}{2}e\phi_3} \cdot e^{i\frac{t}{2}e\phi_2} \cdot e^{i\frac{t}{2}e\phi_1}.
    \end{align}
    The phase operators are similar to the potential operator of the time-dependent Schr\"{o}dinger equation~\cite{benenti2008}. In this case, the electric charge $e$ and the step size $t$ are constant factors, and we can omit the term by rotation angle. Furthermore, we have proven that the remaining term $e^{i\phi_i}$ can be performed with the gate complexity $\mathcal{O}(\poly(\log(1/\varepsilon)))$ using Walsh operators.
\end{proof}

\subsection{Position-dependent coin operator as the vector potential term}
\begin{figure*}[t]
    \centering
    \includegraphics[width=\textwidth]{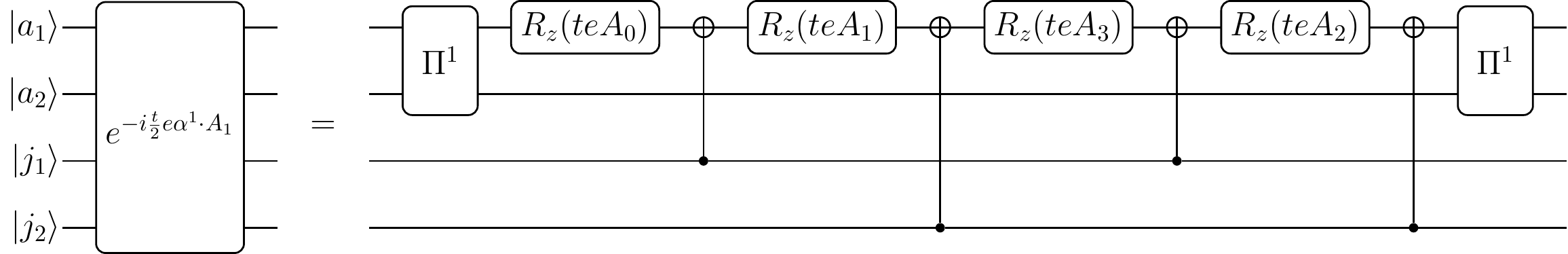}
    \caption{Quantum circuit for implementing an arbitrary vector potential. The sequence of $R_z$ and CNOT gates, ordered according to the Gray code, exactly implement the Walsh operator on two qubits and auxiliaries. $\Pi^1$ gates are applied to transform the $\beta$ basis to the Dirac basis $\alpha^1$ for the $x$ axis.}
    \label{fig:dtqw:vector}
\end{figure*}
This subsection explains the implementation of the space-dependent vector potential term.
\begin{lemma}\label{lem:lemma4}
    For any $t\in\mathbb{R}$, the quantum operation $e^{-i\frac{t}{2}e\bm{\alpha}\cdot\bm{A}}$ can be approximated within error tolerance $\varepsilon$ and with gate complexity $\mathcal{O}(\poly(\log(1/\varepsilon)))$.
\end{lemma}
\begin{proof}
    We first decompose the three-dimensional vector potential term via the product formula, as the Dirac matrices $\alpha^i$ do not commute.
    This decomposition introduces an error $\varepsilon$ up to the 1st-order, similarly to the kinetic term in Sec.~\ref{sec:kinetic_term},
    \begin{multline}
        e^{-i\frac{t}{2}e\bm{\alpha}\cdot\bm{A}}
        = e^{-i\frac{t}{2}e\alpha^3\cdot A_3} \cdot e^{-i\frac{t}{2}e\alpha^2\cdot A_2} \cdot e^{-i\frac{t}{2}e\alpha^1\cdot A_1} \\
        +\mathcal{O}(t^2/r^2).
    \end{multline}
    Again, we can take $r=\mathcal{O}((\norm{H}t)^2/\varepsilon)$ to ensure the error stays at most $\varepsilon$. The vector potential term only differs in sign-flip and the existence of the Dirac matrices $\alpha^i$ from the scalar potential term. Thus we diagonalize the matrices by $\Pi$ matrices defined in~(\ref{eq:Pi_Matrix}), and we get a position-dependent coin operator,
    \begin{align}
        e^{-i\frac{t}{2}e\alpha^1\cdot A_1}=(\Pi^1\otimes\mathbbm{1})\cdot e^{-i\frac{t}{2}e\beta\cdot A_1}\cdot(\Pi^1\otimes\mathbbm{1}),
    \end{align}
    for $x$ axis. We denoted $\mathbbm{1}$ as an identity matrix that only acts on the main register (but not on auxiliary registers). We now conclude that the space-dependent vector potential term can be performed with gate complexity $\mathcal{O}(\poly(\log(1/\varepsilon)))$ using Walsh operators.
\end{proof}

The circuit shown in \fig{dtqw:vector} demonstrates implementing a one-dimensional vector potential energy term with $n=4$. The $R_z$ rotations are ordered according to the Gray code to create the Walsh operator. It should be noted that CNOT gates can be optimized by eliminating $R_z(0)$ rotations.

\begin{theorem}[Digital quantum simulation of the time-dependent Dirac equation with $k$th-order product formula]\label{thm:Theorem}
    Consider an instance of the Dirac equation of the form~(\ref{eq:Schroedinger}) with Hamiltonian defined in~(\ref{eq:Dirac_Hamiltonian}). Then, for any $t\in\mathbb{R}$ (propagation time) and $\varepsilon$ (error tolerance), there exists a quantum algorithm producing a normalized state that approximates $\psi(T)$ with asymptotic gate complexity,
    \begin{align}
        \widetilde{\mathcal{O}}\Bigl( 5^{2k} (\norm{\phi}_{\max}+\norm{\bm{A}}_{\max})^{1+1/2k}T^{1+1/2k}/\varepsilon^{1/2k} \Bigr).
    \end{align}
\end{theorem}
\begin{proof}
    We follow the procedure according to the Ref.~\cite{Childs2022quantumsimulationof}. The amount of quantum Fourier transforms (QFTs) is equivalent to the number of exponentials, and according to the upper bound presented in~\cite[Theorem 1]{Berry_2006}, the number of exponentials is limited by $m=7$ (as the scalar potential terms commute), such that
    \begin{align}
        N_{\mathrm{exp}} \leqslant 14 \cdot 5^{2k} (7\|H\|T)^{1+1/2k}/(\varepsilon/7)^{1/2k},
    \end{align}
    for the error at most $\varepsilon/7$.
    Since
    \begin{align}
        \begin{split}
            \norm{H_1} &= n\pi/\Omega, \\
            \norm{H_2} &= m, \\
            \norm{H_3} &= e\norm{\phi}_{\max}, \\
            \norm{H_4} &= e\norm{\bm{A}}_{\max},
        \end{split}
    \end{align}
    this implies that
    \begin{equation}
        \begin{aligned}
            & N_{\mathrm{exp}} \leqslant 14\cdot5^{2k}7^{1+1/2k} \times \\
            & \Big[T\Big(\frac{n\pi}{\Omega}+m+e\norm{\phi}_{\max}+e\norm{\bm{A}}_{\max}\Big)\Big]^{1+1/2k}/\varepsilon^{1/2k}.
        \end{aligned}
    \end{equation}
    After rescaling the $2k$th-order product formula to the $k$th-order, we have
    \begin{align}
        \norm{\psi(T)-\widetilde{\psi}(T)} \leqslant \varepsilon,
    \end{align}
    where $\widetilde{\psi}(T)=\left[ U_{2k}(t) \right]^r \psi(0)$. The result follows immediately from \lem{lemma1}, \lem{lemma2}, \lem{lemma3} and \lem{lemma4} that $\widetilde{\psi}(T)$ can be implemented with poly-logarithmic gate complexity.
\end{proof}

In comparison with the gate complexity presented in~\cite{Childs2022quantumsimulationof}, which is expressed as $\widetilde{O}( 5^{2k} d(d+\norm{f}_{\max})^{1+1/2k}T^{1+1/2k}/\epsilon^{1/2k} )$, the analysis above demonstrates that the potential terms dominate the asymptotics in both relativistic and non-relativistic real-space dynamics. Note that we considered $d=3$ in the relativistic case.
\section{Applications}
\label{sec:dtqw:simulation results}
\begin{figure*}[t]
    \centering
    \subfloat[]{\includegraphics[width=0.5\textwidth]{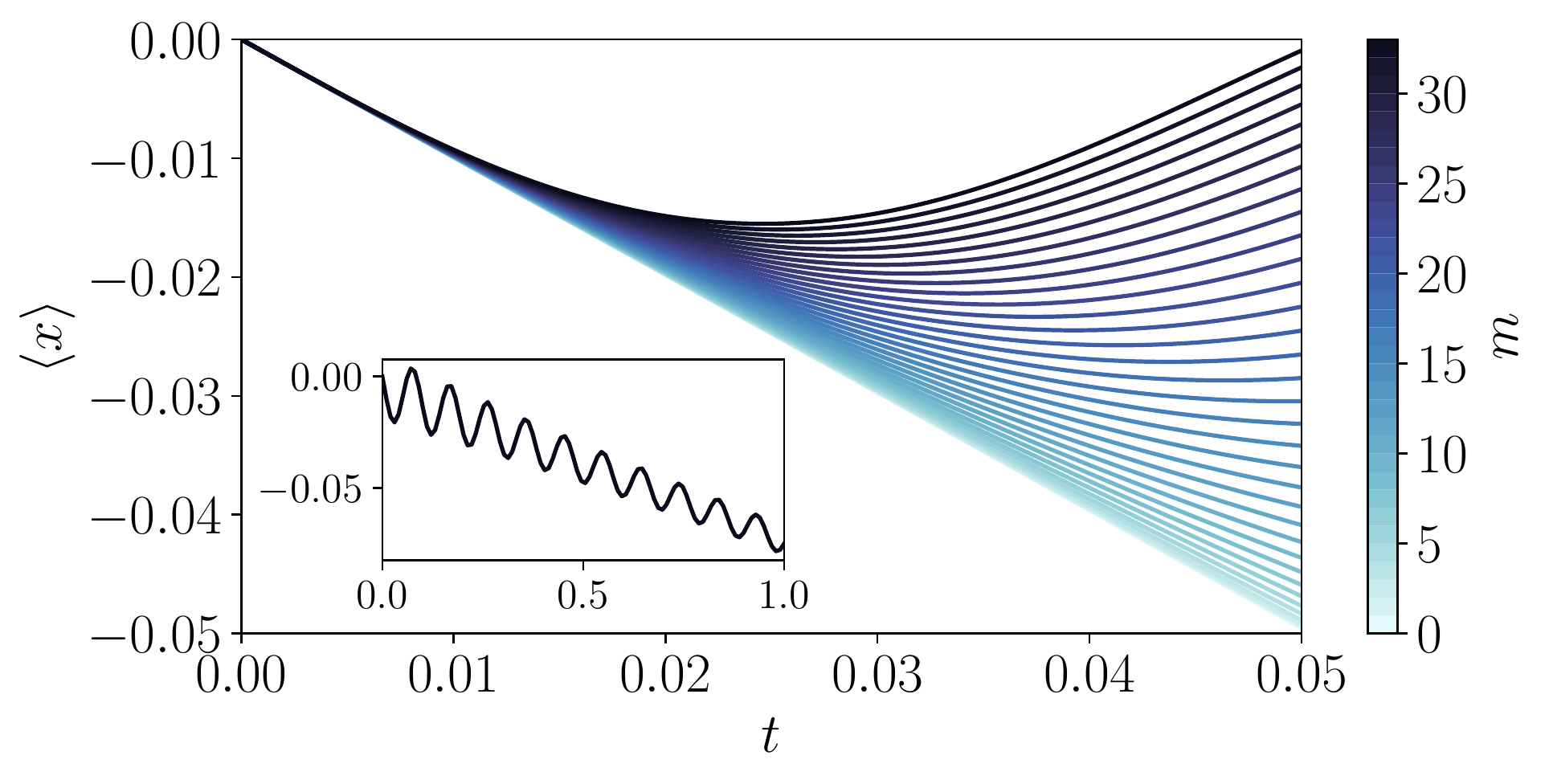}
    \label{fig:results:zitter}}
    \subfloat[]{\includegraphics[width=0.5\textwidth]{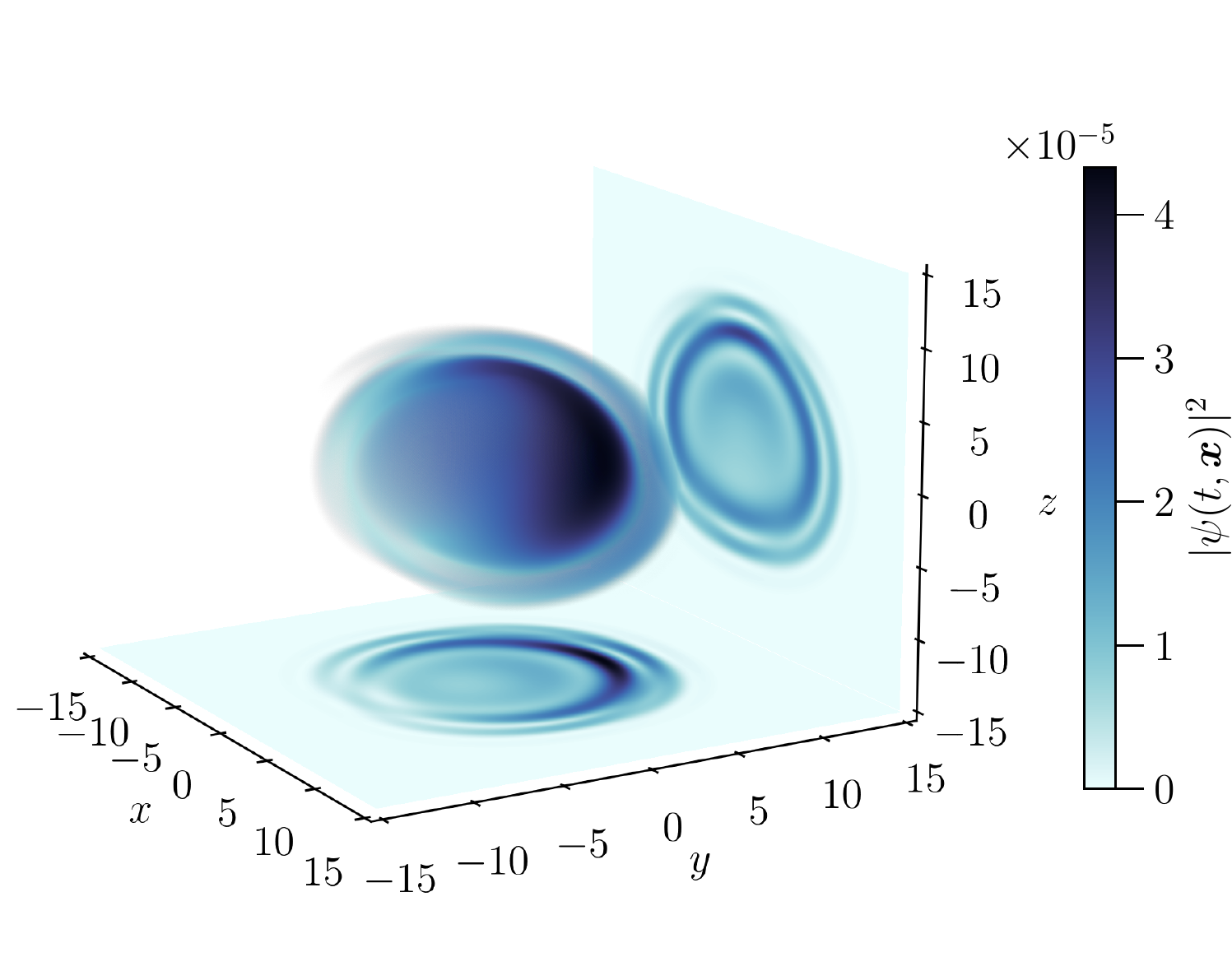}
    \label{fig:results:zitter3d}}
    \caption{(a) Simulation of the Zitterbewegung in 1+1 dimensions. The expectation value $\ev{x}$ of the Dirac particle exhibits increasingly strong trembling motion with higher mass $m$. At rest mass ($m=0$), the Dirac particle moves at the speed of light. The inset shows the long-term dynamics of the particle for the case of $m=33$. (b) Simulation in 3+1 dimensions. An initial spinor distributed around $x,y,z=0$ with particle mass $m=1$ shows the spinor's rotation at time $T=1$, clearly induced by the mass.}
\end{figure*}
This section demonstrates the digital quantum simulator with different settings. We first simulate Zitterbewegung, a phenomenon in relativistic quantum mechanics where a free Dirac particle exhibits a trembling motion even in the absence of any potential terms. This motion is known to arise from the interference between positive and negative energy solutions~\cite{Hestenes1990-ld}. We first consider a one-dimensional massive free Dirac particle to have the interference, assuming $A=\phi=0$ with $\hbar=c=1$. The initial spinor is defined as
\begin{align}
    \psi_0(x)=
    \left(\frac{1}{2\pi\sigma^2}\right)^{1/4}\left(
    \begin{array}{c}
         1  \\
         -1 
    \end{array}\right)
    e^{ip_0x}e^{-\frac{x^2}{4\sigma^2}},
\end{align}
with $p_0=1/4$ and $\sigma=0.05$. The initial momentum distribution is well localized in the momentum space to contribute to the Zitterbewegung. \fig{results:zitter} shows the expectation value of the position in the total time evolution $0.05\ \mathrm{n.u.}$ divided by 100 time-steps with different mass parameters ranging from $m=0$ to $m=33$. When the mass is $m=0$, the corresponding velocity is the speed of light. On the other hand, as the mass increases, the Dirac particle starts to move laterally, leading to a trembling motion.

We test our simulator using the time-dependent Dirac equation in 3+1 dimensions to showcase the three-dimensional Zitterbewegung. We observe the massive Dirac equation in three spatial dimensions in the simulation results displayed in \fig{results:zitter3d}. The figure depicts an instance of the Zitterbewegung at time $t=1$. As time goes on, the initial spinor rotates and propagates. The probability density $\abs{\psi(t,\bm{x})}^2$ is projected on the $x$--$y$ and $y$--$z$ plane to show the orbitals. The truncated space $\bm{x}\in[-15,15]^3$ was evaluated using a total of $n^3=(2^8)^3$ grid points.

We study the potential term by simulating the Klein paradox~\cite{Klein1929-xu,1999PhR...315...41D,2004PhRvL..92d0406K,Chiral2006,1993JPhA...26.1001S}, which considers high-energy particles scattering from a potential barrier.
In non-relativistic mechanics, governed by the Schr\"{o}dinger equation, the probability of transmission decreases with the height of the potential step.
This is in stark contrast with the situation in relativistic quantum mechanics, where the incoming particle's chance of getting transmitted increases with the height of the potential.
The incoming positive energy states can tunnel into the barrier and appear as negative energy states.

We use the Walsh operators to implement a scalar step potential barrier and place it at the origin $z=0$.
This choice results in the reduction of the number of CNOT gates needed to implement the potential step.
After optimization, the gates eventually become a single $R_z$ gate,
\begin{equation}
    \begin{minipage}[c]{.4\textwidth}
        \centering
        \begin{quantikz}[draw]
            \lstick{$e^{-i\frac{t}{2}V_0}=\ket{j_1}$} & \gate{R_z(tV_0)} & \qw
        \end{quantikz},
    \end{minipage}
\end{equation}
because only one Walsh coefficient remains from the step potential. Note that the number of coefficients depends on the potential function, which a discrete Walsh transform can evaluate, and the number of CNOT gates increases if the position of the potential step changes. Additionally, the potential height depends on the number of rotations, and $V_0$ determines the potential height.
We set the initial wave packet according to~\cite{PhysRevA.72.064103}
\begin{align}\label{eq:Eigenstate}
\psi_0(z)=\mathcal{N} P^{(p)}_+ e^{ipz_{0}} e^{-(p-p_0)^{2} \delta z^{2}}
\end{align}
where $\mathcal{N}$ is the normalization constant, $\delta z=0.03 \ \mathrm{a.u.}$ is the initial spatial width, and $p_0=106.4 \ \mathrm{a.u.}$ is the central momentum.
Moreover, the projector $P^{(p)}_+$ constructs a pure positive energy spinor,
\begin{align}
    P^{(p)}_+=\frac{1}{2}\biggl(\mathbb{I}_2+\frac{c\alpha^1p+\beta mc^2}{\sqrt{c^2p^2+m^2c^4}}\biggr),
\end{align}
with one-dimensional Dirac matrices $\alpha^1=\sigma^1$ and $\beta=\sigma^3$.

\begin{figure*}[t]
    \centering
    \includegraphics[width=\textwidth]{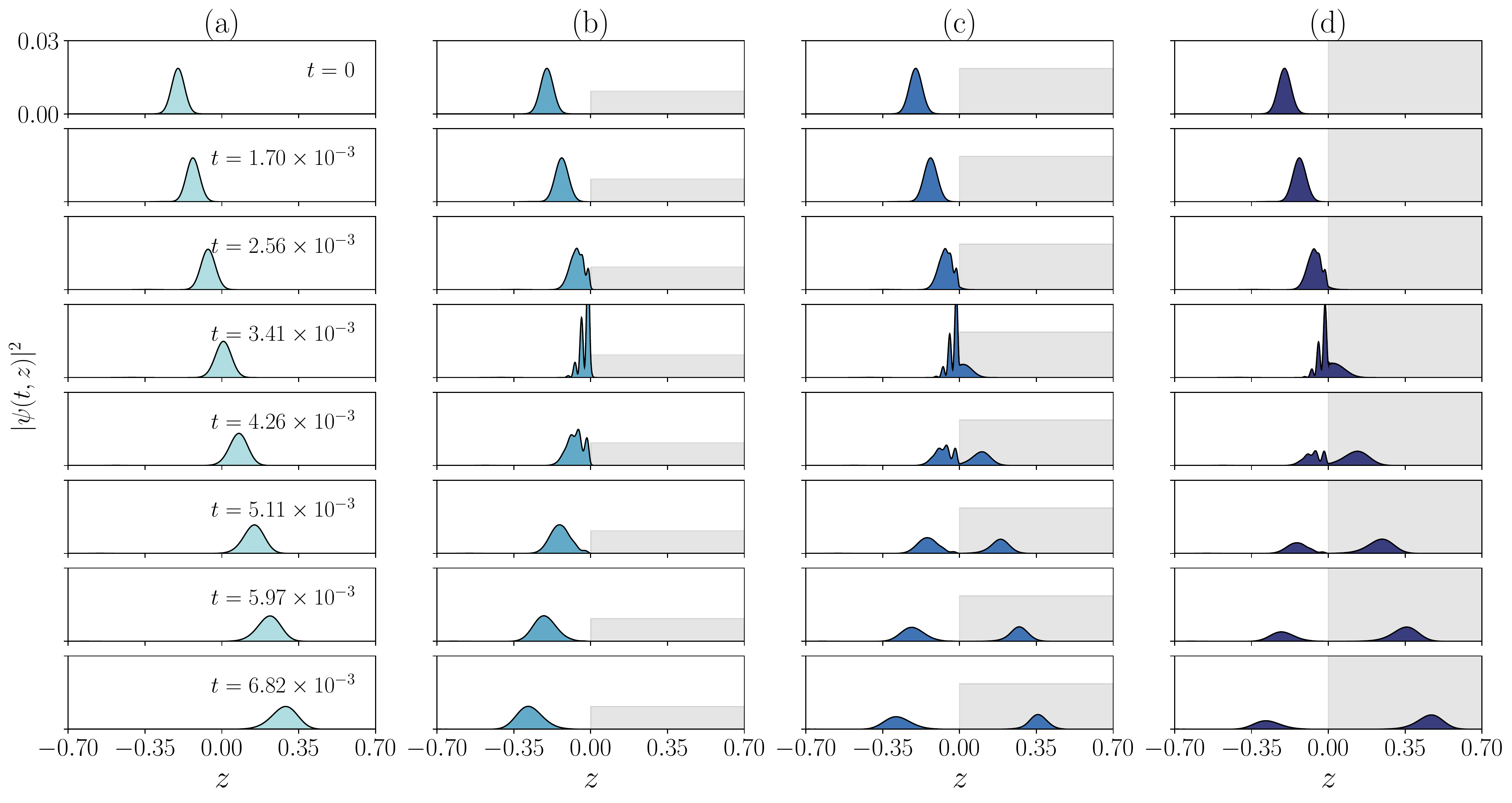}
    \caption{Simulation of the Klein paradox in 1+1 dimensions for increasing size of the step potential, shown by the shaded area. (a) The potential step is $V_0=0$, and the electron propagates freely. (b) $V_0=\Omega mc^2$, and the electron is fully reflected. (c) $V_0=2\Omega mc^2$, and the Klein pair production is observed, with the positron having a finite probability of transmission. (d) $V_0=4\Omega mc^2$ shows a further increase in the probability of transmission.}
    \label{fig:Klein_Paradox}
\end{figure*}

In \fig{Klein_Paradox}, we show the time evolution of the positive momentum eigenstates as an initial wave packet divided into 512 steps for the total simulation time of $T=6.82\times10^{-3} \ \mathrm{a.u.}$.
We consider four scenarios, the first one being the propagation of a free particle.
We observe that the particle propagates in the positive $z$-direction unimpeded, as shown in \figg{Klein_Paradox}{a}.
Next, we repeat the simulation with the same incident particle but introduce a potential step at $z=0$, depicted by the shaded area in \figg{Klein_Paradox}{b}.
The potential step is given by $V_0=\Omega mc^2$, where $\Omega=1.4$ is the computational cell volume.
Klein paradox predicts that this situation should result in a total reflection of the incoming particle, which is precisely what we observe in our simulation.
In~\figg{Klein_Paradox}{c} and~\figg{Klein_Paradox}{d}, the potential step is increased further to $V_0=2\Omega mc^2$ and $V_0=4\Omega mc^2$, respectively.
We observe that despite the increased height of the step potential, the incident particle has an increasing probability of transmission in accordance with the predictions of the Klein paradox.
Since the step potential is located in the center, at the moment the distribution of the packet on the right-hand side reached the center, it hit the potential wall, and the interference pattern appeared. When the wave packet transmits to the right region, Klein pair production occurs so that it transforms into the antiparticle.

\section{Conclusions}
\label{sec:dtqw:conclusions}

In this work, we have introduced a digital quantum simulator based on discrete-time quantum walks, which can efficiently simulate the time-dependent Dirac equation in 3+1 dimensions. We have demonstrated that each Trotterized operator can be implemented using simple shift (position-dependent) coin and phase operators of discrete-time quantum walks. By utilizing a sequence of single-qubit rotations and CNOT gates, we have shown that the Walsh operators for arbitrary scalar and vector potentials can be constructed. We have also developed quantum spectral methods that can approximate the solution globally to simulate relativistic dynamics and converge exponentially faster than ordinary finite-difference methods for smooth solutions. Moreover, we have analyzed the gate complexity of the overall simulation.

We used quantum spectral methods to simulate the 3+1 dimensional time-dependent Dirac equation with scalar and vector potential terms. Our analysis of gate complexity revealed that the increase in complexity is negligible for the transition from non-relativistic to relativistic systems. It strongly depends on the order of the product formula. Moreover, we found that the complexity can be improved for potential functions with appropriate structures, as the potential terms usually dominate the complexity. Our simulation of the Klein paradox with step potential demonstrated that it could be implemented with a single $R_z$ gate. It is important to note that the simulator requires two auxiliary qubits for 3+1 dimensions. In general, the number of auxiliary qubits scales logarithmically with the dimension.

To the best of our knowledge, this study presents the first quantum simulator for the time-dependent Dirac equation with poly-logarithmic complexity. While there have been several previous attempts at developing a quantum simulator for the time-dependent Schr\"{o}dinger equation~\cite{benenti2008,Childs2022quantumsimulationof}, there have been few for the Dirac equation. For example, a recent work by Fillion-Gourdeau \emph{et al.}~\cite{Fillion_Gourdeau_2017} proposed quantum walk algorithms for the Dirac equation using multi-controlled CNOT gates to propagate a walker (solution) in the position space. In contrast, our study proposes a unique alternative scheme for the simulation that employs the quantum Fourier transform to evaluate the kinetic term, giving a global approximation of the solution in the momentum space. This approach mitigates errors and reduces gate complexity for a smoother solution. Our study stands out by explicitly showing the feasibility of this approach.

One limitation of our study is that it assumes the availability of an efficient state preparation procedure for the simulator. Although we have developed and analyzed time-evolution algorithms, initializing the quantum state on quantum registers must be done efficiently. This can be achieved using amplitude encoding algorithms~\cite{PhysRevResearch.4.023136,mitsuda2023approximate}, which allow initializing a general wave packet. However, in most cases, we are interested in starting from the ground state of a given Hamiltonian. Several algorithms have been proposed for this case, such as the imaginary time evolution algorithms~\cite{leadbeater2023nonunitary}. While evolution becomes non-unitary in imaginary time evolution, which is not executable on a quantum computer, recent studies have shown that Hamiltonian simulation techniques can make it executable. Therefore, combining other state preparation algorithms could potentially overcome this limitation.

Our findings have significant implications for the reduced computational cost of simulating the Dirac equation, making it possible to have accurate quantum simulations with relativistic corrections.
Recent advances in quantum walks~\cite{PhysRevA.88.042301,Arrighi2016,mallick2019simulating} and analog quantum simulators~\cite{Viermann_2022} suggest that quantum computers can simulate the Dirac equation in curved spacetime~\cite{Pollock:2010zz}.
We believe that our algorithm represents a promising starting point in incorporating curved spacetime metrics into the simulation.
In addition, recent studies have shown that quantum walks can be extended to quantum cellular automata~\cite{PhysRevA.102.062222}, where the model corresponds to quantum field theories in the continuum limits. Incorporating these extensions into our simulator is an inspiring direction for future research.

\begin{acknowledgments}
S.M. is supported by MEXT Quantum Leap Flagship Program Grant Number JPMXS0120319794 and the National Institute of Information and Communications Technology (NICT Quantum Camp).
T.S. is supported by MEXT KAKENHI Grant Number 22K19781.
\end{acknowledgments}

\bibliographystyle{apsrev4-2}
\bibliography{refs}

\end{document}